\documentclass[english, a4paper]{article}

\usepackage[english]{babel}
\usepackage[T1]{fontenc}
\usepackage{authblk}

\usepackage[a4paper,top=3cm,bottom=3cm,left=3cm,right=3cm,marginparwidth=1.75cm]{geometry}

\usepackage[colorlinks=true, allcolors=blue, hypertexnames=false]{hyperref}
\usepackage{empheq}
\usepackage{graphicx}
\usepackage[colorinlistoftodos]{todonotes}
\usepackage{amsmath}
\usepackage{mathtools}
\usepackage{upgreek}
\usepackage{amssymb}
\usepackage{amsfonts}
\usepackage{amsthm}
\usepackage{hyperref}
\usepackage{enumitem}
\usepackage{mathrsfs}
\usepackage{fontenc}
\usepackage{verbatim}
\usepackage{framed}
\usepackage{algorithm}
\usepackage{bbm}
\usepackage{algpseudocode}
\usepackage{appendix}
\usepackage{color}
\usepackage{multirow}

\usepackage{multirow}
\usepackage{booktabs}
\usepackage{caption}
\usepackage{diagbox}
\usepackage{rotating}
\usepackage{pdflscape} 

\usepackage[normalem]{ulem}

\numberwithin{equation}{section}
\theoremstyle{plain}
\newtheorem{thm}{Theorem}[section]

\newtheorem{prop}{Proposition}

\newtheorem{assu}{Assumption}
\newtheorem{defn}{Definition}[section]
\newtheorem*{rem}{Remark}

\theoremstyle{definition}

\theoremstyle{remark}

\title{Prediction of high-frequency futures return directions based on the mean uncertainty classification methods: An application in China's future market}

\author[1]{Ying Peng}
\author[1]{Yifan Zhang}
\author[1]{Xin Wang}
\affil[1]{Shandong University-Zhongtai Securities Institute for Financial Studies, Shandong University, 250100, China}
\begin{document}
\maketitle

\begin{abstract}

In this paper, we mainly focus on the prediction of short-term average return directions in China's high-frequency futures market. As minor fluctuations with limited amplitude and short duration are typically regarded as random noise, only price movements of sufficient magnitude qualify as statistically significant signals. Therefore data imbalance emerges as a key problem during predictive modeling. From the view of data distribution imbalance, we employee the mean-uncertainty logistic regression (mean-uncertainty LR) classification method under the sublinear expectation (SLE) framework, and further propose the mean-uncertainty support vector machines (mean-uncertainty SVM) method for the prediction. Corresponding investment strategies are developed based on the prediction results. For data selection, we utilize trading data and limit order book data of the top 15 liquid products among the most active contracts in China's future market. Empirical results demonstrate that comparing with conventional LR-related and SVM-related imbalanced data classification methods, the two mean-uncertainty approaches yields significant advantages in both classification metrics and average returns per trade.

\end{abstract}

\section{Introduction}\label{sec:intro}

‌As an important component of the global financial system, futures markets have experienced rapid development over the past two decades. Recent advances in information technology and computational capabilities facilitated high-frequency futures trading, enabling future transactions to be executed on a second or even millisecond timescale. Since high-frequency data not only elucidates market microstructure but also captures traders' behavioral patterns and real-time market reactions, how to effectively leverage such data for prediction has become a critical challenge in high-frequency futures trading research. Consequently, predicting the direction of market price movement based on high-frequency futures trading data stands as a prominent research focus.

Assets prices often exhibit strong persistence and predictability over ultra-short horizons, meaning that the direction of price movements can be forecasted within short time intervals. Fan et al. \cite{Fan2022how} found that the predictability of returns is much more systematic and significant under ultra-short-term and high-frequency conditions comparing with low-frequency data. They demonstrated that the price movement directions of nearly all stocks is highly predictable over ultra-short periods, and validated this conclusion by implementing machine learning algorithms with high-frequency trading data. Existing studies \cite{andersen2005correcting,baumeister2015high,lyocsa2021stock,ye2024short,peng2024attention, song2024modelling} have also confirmed the predictive power of high-frequency data from various perspectives.

In high-frequency futures trading markets, minor fluctuations with limited amplitude and short duration are typically regarded as random noise, failing to reflect actual market trends. ‌Consequently, only price movements of sufficient magnitude qualify as statistically significant signals.‌ Since such significant signals are relatively few, data imbalance emerges as a key problem during predictive modeling, necessitating specialized techniques.  Current methodologies for imbalanced data classification primarily fall into two different categories:‌ sampling-based methods and algorithmic optimization techniques.‌ ‌On the sampling front‌, conventional strategies include random undersampling (RUS), synthetic minority oversampling (SMOTE), and their hybrid variants \cite{chawla2002smote,kubat1997addressing,vairetti2024efficient,sun2024nearest,sun2024undersampling}. ‌In the algorithmic domain‌, researchers focus on structural modifications of models or loss function adaptations to better accommodate minority class samples \cite{alvim2010daily,chen2024cost,chamlal2024hybrid,guo2024adaptive}.
Nevertheless, while both methodological approaches seek to improve classification accuracy through either sample quantity manipulation or weight adjustment, they fail to adequately address inherent uncertainty of the data distribution.

Recent research indicates that nonlinear expectation theory \cite{peng1997backward,peng2004filtration,peng2005nonlinear,peng2010nonlinear} can effectively characterize distribution uncertainties of random variables. Different from the classical linear expectation framework \cite{kolmogorov2018foundations}, a nonlinear expectation represents a ‌family‌ of potential probability distributions. In 2007, Peng \cite{peng2007g}  further introduced the ‌sublinear expectation (SLE) theory‌. Under the SLE framework, the ‌distribution uncertainty‌ of random variables primarily reflects in two aspects: ‌mean uncertainty‌ and ‌volatility uncertainty‌. With the development of the theoretical research, the application significance of the SLE framework has grown increasingly apparent. Hu and Ji {\cite{hu2016stochastic,hu2017dynamic}} derived the maximum principle and dynamic programming principle under volatility uncertainty by solving stochastic optimal control problems. Based on the law of large numbers and the central limit theorem under nonlinear expectations, Peng and Jin \cite{jin2016optimal} proposed a max-mean estimation approach to obtain optimal unbiased estimators under distribution uncertainty. In the application area, Peng and Yang \cite{peng2023improving} developed a novel Value-at-Risk (VaR) predictor under model uncertainty, which demonstrates notable advantages over most existing benchmark VaR methods. In the context of regression analysis, Yang and Yao \cite{yang2021linear} investigated robust linear regression models where the mean and variance of covariates are uncertain. They proposed two estimation strategies for the regression coefficients under a given loss function. Related work includes that of Xu and Xuan \cite{xu2019nonlinear} who approached the problem from an algorithmic perspective and replaced the classical squared loss function with a more robust max-mean loss function. Recently, the SLE theory has also demonstrated advantages in the field of ‌imbalanced data classification‌. Ji et al. \cite{ji2023imbalanced} proposed ‌mean-uncertain logistic regression (mean-uncertainty LR)‌ and ‌volatility-uncertain logistic regression (volatility-uncertainty LR)‌ methods for binary classification under the SLE. Compared to conventional ‌sampling-based methods for imbalanced classification‌, these two approaches have exhibited superior performance across multiple datasets, and the mean-uncertain LR method demonstrate particularly significant classification efficacy.

In order to effectively predict market price movement directions based on high-frequency trading data in China's futures market, we apply and extend the method of Ji et al. \cite{ji2023imbalanced} for the predictions of short-term average return directions under the SLE framework. First, extended from the mean-uncertainty LR method of \cite{ji2023imbalanced}, we develop the ‌mean-uncertainty support vector machine (mean-uncertainty SVM) approach and provide theoretical proof of its property. Subsequently, we apply both the ‌mean-uncertainty LR and the ‌mean-uncertainty SVM models to China's high-frequency futures market for the prediction of short-term average return directions, then develop corresponding investment strategies based on the prediction results. Specifically, two distinct binary classification tasks are formulated for the prediction: the first predicts "upward movement" versus "non-upward" states, while the second targets "downward movement" versus "non-downward" states. In the empirical analysis, we choose the top 15 futures products with higher liquidity in the most active contracts of China's futures market, using the corresponding trading data and limit order book data as the original dataset. Empirical results demonstrate that both the two mean-uncertainty methods outperform conventional sampling-based classification approaches in predictive performance. From the view of average returns per trade, both methods surpass conventional sampling-based classifiers on ‌80\% of futures products for "upward / non-upward" predictions; for "downward/non-downward" predictions, the mean-uncertainty LR achieves higher returns on ‌80\% of products‌, and the mean-uncertainty SVM shows superior performance on ‌67\% of products‌.

\section{Methodology}\label{sec:theo}

In this section, we mainly present the two mean-uncertainty classification methods employed in this paper for the prediction of short-term return directions on high-frequency futures. Firstly, we briefly introduce the idea of the mean-uncertainty LR method of Ji et al.~\cite{ji2023imbalanced} under the SLE framework. Then inspired by \cite{ji2023imbalanced}, we further develop the mean-uncertainty SVM method, which can be viewed as a correction of SVM method under model uncertainty. The related definitions and theorems SLE can be seen in the Appendix.

\subsection{Mean-uncertainty LR}\label{subsec:MU_LR}

Under the framework of SLE, Ji et al. \cite{ji2023imbalanced} addressed the imbalanced classification problem from the view of model uncertainty and developed two improved  logistic regression (LR) method: the mean-uncertainty LR and volatility-uncertainty LR methods. Their empirical results demonstrate that the mean-uncertainty LR outperforms conventional imbalanced data classification methods across different datasets.

Motivated by the above research results, we employ the mean-uncertainty LR approach to predict short-term average return directions, where only very significant return fluctuations are considered to be "upward" or "downward". Then the investment policy can be developed according to the prediction results. For example, when the average return directions are predicted to be "upward", the investor can buy a certain number of futures.

Before introducing the mean-uncertainty LR method, we introduce the standard LR model via the classical linear regression framework. Given a dataset $\left\{ \left( {\boldsymbol{x}_{i}},Y_{i}^{*} \right) \right\}_{i=1}^{n}$, consider the following linear regression model:
\begin{equation*}
Y_{i}^{*}=\boldsymbol{x_{i}'\beta_{1}} + \beta_{0} +\delta_{i} ,i=1,2,...,n,
\end{equation*}
where $\boldsymbol{x_{i}}\in \mathbb{R} ^{n}$ is the covariate vector, $Y_{i}^{*}\in \mathbb{R}$ is a continuous latent response variable, and $\delta_{i}$ is an independently and identically distributed (i.i.d.) disturbance term. $\boldsymbol \beta=(\boldsymbol \beta_{1}',\beta_{0})' \in \mathbb{R} ^{n+1}$ is a vector of parameters to be estimated. In practical binary classification scenarios, it is often the case that the values of the latent variable \( Y_{i}^* \) are not directly observable. Instead, what we can observe are the binary classification labels \( Y_{i}\), such as $Y_i = \begin{cases}
0, & \text{if } Y_i^* < 0, \\
1, & \text{if } Y_i^* \ge 0.
\end{cases}$ Without loss of generality, we assume that label "1" denotes the positive class and label "0" the negative class.

The LR model links the latent response variable $Y_i^*$ to the binary observation $Y_i$ via the sigmoid function. Under the classical linear i.i.d. assumptions, when the disturbance term $\delta_i$ follows a standard logistic distribution, $Y_i|\boldsymbol{x}_i$ follows a degenerate Bernoulli distribution where $P(Y_i=1|\boldsymbol{x}_i)$ is a deterministic probability. Define $z_i = \boldsymbol{x}_i' \boldsymbol{\beta_{1}} + \beta_{0}$, \(\sigma(z) = \frac{1}{1 + e^{-z}}\), the conditional probability for \(\{Y_i = 1\}\) is:
\begin{equation*}
P(Y_i=1|\boldsymbol{x_i}) = E[Y_i|\boldsymbol{x_i}]=\sigma(z) = \frac{1}{1 + e^{-\boldsymbol{x_{i}'}\boldsymbol{\beta_1}+{\beta_0}}}.
\end{equation*}
Observations are classified as positive (${Y}_i=1$) when $P(Y_i=1 \mid \boldsymbol{x}_i) > 0.5$, and negative otherwise.

Under the SLE framework, introducing the mean-uncertainty in $\delta_i$ implies that $Y_i|\boldsymbol{x}_i$ is governed by a family of Bernoulli distributions, which generalize classical linear expectation models and better captures real-world data characteristics.

\begin{assu}\label{ass2.1}
Assume that the disturbance term \(\delta_i\) satisfies \(\delta_i = \varepsilon_i + M_i\), where \(\{ \varepsilon_i, M_i \}_{i=1}^{n}\) is an independent and identically distributed (i.i.d.) sequence under the SLE. \(M_i\) follows the maximum distribution \(M_{[\underline{\mu}, \overline{\mu}]}\), and \(\varepsilon_i\) is nonlinearly independent of \(M_i\). In particular, \(\varepsilon_i\) follows the standard logistic distribution.
\end{assu}

\begin{prop}\label{pro2.1}
Under Assumption~\ref{ass2.1}, given the parameters of the LR model $\boldsymbol \beta=(\boldsymbol \beta_{1}',\beta_{0})'$ and the sample \(\boldsymbol{x_i}\), the probability uncertainty of \(\{Y_i = 1\}\) can be represented by the uncertainty interval \(\left[ \underline{P_i}, \overline{P_i} \right]\), where:

\[
\underline{P_i} = \sigma \left( \boldsymbol{x_i}^{\prime} \boldsymbol{\beta_1} + \beta_0 + \underline{\mu} \right), \quad \overline{P_i} = \sigma \left( \boldsymbol{x_i}^{\prime} \boldsymbol{\beta_1} + \beta_0 + \overline{\mu} \right),
\]
and \(\sigma(x) = \frac{1}{1 + e^{-x}}\).

\end{prop}

Under Assumption~\ref{ass2.1}, if \( \boldsymbol{\hat{\beta}} \) is the parameter estimate obtained from the initial logistic regression model, then according to Proposition ~\ref{pro2.1}, the sample should satisfy the following moment condition:

\begin{align*}
   \widehat{\mathbb{E}}[{{Y}_{i}}-\sigma (\boldsymbol{x_{i}^{'}}\boldsymbol{{{\hat{\beta }}}_{1}}+{{{\hat{\beta }}}_{0}}+{\overline{{\mu }} })]&=0, \\
  -\widehat{\mathbb{E}}[-\left( {{Y}_{i}}-\sigma (\boldsymbol{x_{i}^{'}}{\boldsymbol{{\hat{\beta }}}_{1}}+{{{\hat{\beta }}}_{0}}+\underline{{\mu }} ) \right)]&=0.
\end{align*}
Where \( \overline{\mu} \) and \( \underline{\mu} \) are the error distribution parameters to be estimated. Using the method of moment estimation, \( n - N \) error statistics are constructed:
\begin{align*}
    Z_{N}^{k} :=\; & \frac{1}{N}\sum_{j=1}^{N}\Bigl[ Y_{k+j-1} - \sigma\Bigl(\boldsymbol{ x_{k+j-1}'\hat{\beta}_{1}} + \hat{\beta}_{0} + \overline{\mu} \Bigr) \Bigr],k=1,\cdots,n-N.
\end{align*}

According to Theorem~\ref{the2.1}, when \( N \) is sufficiently large, \( Z_{N}^{k} \) converges in distribution to a maximum distribution whose upper mean parameter is $\widehat{\mathbb{E}}[Y_i - \sigma (\boldsymbol{x_i^{'} \hat{\beta}_1} + \hat{\beta}_0 + \overline{\mu})]$,
then we have
\begin{align*}
    \underset{N\to \infty }{\mathop{\lim }}\,\widehat{\mathbb{E}}\left[ Z_{N}^{k} \right]=\widehat{\mathbb{E}}[{{Y}_{i}}-\sigma (\boldsymbol{x_{i}^{'}{{\hat{\beta }}_{1}}}+{{\hat{\beta }}_{0}}+\overline{{\mu }} )].
\end{align*}

By Theorem~\ref{the2.3}, the sample statistic \( \underset{1 \le k \le n - N}{\max} Z_{N}^{k} \) is the asymptotically optimal unbiased estimate of the population moment \( \widehat{\mathbb{E}}[Y_i - \sigma (\boldsymbol{x_i^{'} \hat{\beta}_1} + \hat{\beta}_0 + \overline{\mu})] \). The empirical version of the moment condition is:
\begin{align*}
    \underset{1\le k\le n-N}{\mathop{\max }}\,\frac{1}{N}\sum\limits_{j=1}^{N}{[{{Y}_{k+j-1}}-\sigma (\boldsymbol{x_{k+j-1}^{'}{{{\hat{\beta }}}_{1}}}+{{{\hat{\beta }}}_{0}}+\overline{\mu })]}=0.
\end{align*}

The probability of \( \{Y_i = 1\} \) can be expressed in terms of the upper probability, which is defined as:

\[
\underset{P \in \mathscr{P} _{\Theta }}{\sup} P(Y_i = 1 | x_i) = \sigma (\boldsymbol{x_i' \beta_1}+ \beta_0 + \overline{\mu}).
\]

In summary, in the method of \cite{ji2023imbalanced}, the approximate estimate ${{\hat{\overline{\mu }}}_{N}}$ of the error distribution parameter is firstly determined by solving a nonlinear equation:
\begin{align}
    \underset{1\le k\le n-N}{\mathop{\max }}\,\frac{1}{N}\sum\limits_{j=1}^{N}{[{{Y}_{k+j-1}}-\sigma (\boldsymbol{{{x}^{'}}{{{\hat{\beta }}}_{1}}}+}{{\hat{\beta }}_{0}}+\mu )]=0\Rightarrow {{\hat{\overline{\mu }}}_{N}},
\end{align}
where \( N \) represents the size of the sliding window, which can be treated as a hyper-parameter of the model and selected using a grid search method.

In the following, when the approximate estimate \( \hat{\overline{\mu}}_N \) of \( \mu \) is obtained, the upper probability is adjusted accordingly. Then we have:

\begin{align}\label{upper P for mean-uncertainty LR}
    \overline{P_i}=\underset{P\in {\mathscr{P} _{\Theta }}}{\mathop{\sup }}\,P({{Y}_{i}}=1|\boldsymbol{{x}_{i}})=\sigma (\boldsymbol{{x}'{{\hat{\beta}}_{1}}}+{{\hat{\beta }}_{0}}+{{\hat{\overline{\mu }}}_{N}}).
\end{align}

If the upper probability $\overline{P_i}$ is greater than 0.5, the sample is classified as the minority class; otherwise, it is classified as the majority class. Similarly, we can estimate $\underline{{\mu}}_N$ of the error distribution parameter and then obtain the lower probability $\underline{P_i}$, see the details in \cite{ji2023imbalanced}. According to Proposition ~\ref{pro2.1}, when $\overline{\mu }=\underline{\mu }$, we have $\overline{P_i}=\underline{P_i}$, then the probability of $Y_i=1$ degenerate to the traditional LR prediction $\sigma \left( \boldsymbol{x_i}^{\prime} \boldsymbol{\beta_1} + \beta_0 \right)$.

\subsection{Mean-uncertainty SVM}\label{subsec:MU_SVM}

Inspired by the work of \cite{ji2023imbalanced}, we further develop the mean-uncertainty SVM classification method, which can be viewed as a correction method of the SVM under SLE. In SVM, the samples are mapped to a high-dimensional space with a kernel function, and the distance from each sample point to the hyperplane is computed as follows:

\begin{align}\label{1}
    d\left( \boldsymbol{x} \right)=\sum\limits_{j=1}^{N}{{{\alpha }_{j}}\tilde{{Y}_{j}}K\left( \boldsymbol{x},\boldsymbol{\tilde{{x}}_{j}} \right)}+\delta,
\end{align}

where $\boldsymbol{\tilde{{x}}_{j}}$ is the test sample to be classified, \(\boldsymbol{\tilde{x}}\) is the training sample, \(\tilde{{Y}_{j}}\in \{0, 1\}\) is the class label of the training sample, \(\alpha_j\) is the Lagrange multiplier, and \(K(\boldsymbol{x}, \boldsymbol{\tilde{{x}}_{j}})\) is the kernel function. The sign of \(d(\boldsymbol{x})\) indicates which side of the hyperplane the test sample lies on: a positive value indicates the positive class, and a negative value indicates the negative class.

Extended from Proposition ~\ref{pro2.1}, the mean-uncertainty SVM method is developed in the following.
\begin{prop}\label{pro2.2}
Let Assumption~\ref{ass2.1} hold. Given the model parameters of the SVM, \(\boldsymbol{\beta} = (\boldsymbol{\beta_1'}, \beta_0)'\) and the sample \(\boldsymbol{x_i}\), the probability uncertainty of \(\{Y_i = 1\}\) can be described by the interval \([ \underline{P_i}, \overline{P_i} ]\):

\[
\underline{P_i} = \sigma(d(\boldsymbol{x_i}) + \underline{\mu}), \quad \overline{P_i} = \sigma(d(\boldsymbol{x_i}) + \overline{\mu}),
\]
where \(\sigma(x) = \frac{1}{1 + e^{-x}}\), and the form of \(d(\boldsymbol{x_i})\) is defined in Equation~\eqref{1}.
\end{prop}

\begin{proof}
Let \(\mathscr{P}_{\Theta}\) be a set of probability uncertainty, and \(Y_i\) be a random variable taking values in \(\{0, 1\}\). Then \(\widehat{\mathbb{E}}[\cdot] = \sup\limits_{P \in \mathscr{P}_{\Theta}} E_P(\cdot)\) is a sublinear expectation. \(Y_i\) can be expressed as
\begin{equation*}
Y_i =
\begin{cases}
0, & \text{if } d(\boldsymbol{x_i}) < 0, \\
1, & \text{if } d(\boldsymbol{x_i}) \geq 0,
\end{cases}
\end{equation*}
according to Assumption~\ref{ass2.1}, \(\delta_i = \epsilon_i + M_i\), where \(M_i\) follows the maximum distribution \(M_{[\underline{\mu}, \overline{\mu}]}\), and \(\epsilon_i\) is nonlinearly independent of \(M_i\). Therefore, we have
\begin{align*}
\sup\limits_{P \in \mathscr{P}_{\Theta}} P(Y_i=1|X_i=\boldsymbol{x_i})
&= \sup\limits_{P \in \mathscr{P}_{\Theta}} P(d(\boldsymbol{x_i})+\delta_i \geq 0) \\
&= \sup\limits_{P \in \mathscr{P}_{\Theta}} E_P\left[\mathbf{1}_{\{d(\boldsymbol{x_i})+\delta_i \geq 0 \}}\right] \\
&= \widehat{\mathbb{E}}\left[\mathbf{1}_{\{d(\boldsymbol{x_i})+\delta_i \geq 0 \}}\right]  \\
&= \widehat{\mathbb{E}}\left[\mathbf{1}_{\{d(\boldsymbol{x_i})+\epsilon_i + M_i \geq 0 \}}\right]  \\
&= \widehat{\mathbb{E}}\left[\sigma(d(\boldsymbol{x_i})+ M_i \geq 0)\right]  \\
&= \sigma(d(\boldsymbol{x_i})+ \overline{\mu})  \\
&= \overline{P_i}.
\end{align*}
Similarly,
    \[
    \underline{P_i}=\inf_{P \in \mathscr{P}_{\Theta }} P(Y_i=1|X_i=\boldsymbol{x_i})
    = \sigma\bigl(d(\boldsymbol{x_i})+\underline{\mu}\bigr).
    \]
\end{proof}

According to Proposition ~\ref{pro2.2}, the classification procedure of the mean-uncertainty SVM method is as follows:
\begin{itemize}
\item step 1: Determine the approximate estimate of the error distribution parameter \( \overline{\mu} \) by solving a nonlinear equation:
\begin{equation}
\underset{1\le k\le n-N}{\mathop{\max }}\,\frac{1}{N}\sum\limits_{j=1}^{N}\Bigl[ Y_{k+j-1} - \sigma\Bigl(d(\boldsymbol{x_i})+\mu\Bigr) \Bigr]=0\Rightarrow \hat{\overline{\mu}}_N,
\end{equation}
where \( N \) represents the size of the sliding window.

\item step 2: Incorporate the approximate estimate \( \hat{\overline{\mu}}_N \) into the calculation of the upper probability, yielding the output of the class probability:

\begin{equation}\label{upper P for mean-uncertainty SVM}
\underset{P\in \mathscr{P}_{\Theta}}{\mathop{\sup }}\,P(Y_{i}=1|\boldsymbol{x_{i}})=\sigma\Bigl(d(\boldsymbol{x_i}) + \hat{\overline{\mu}}_N\Bigr),
\end{equation}

if the class probability is greater than 0.5, the sample is classified as the minority class, i.e., class "1"; otherwise, it is classified as the majority class, i.e., class "0".
\end{itemize}

\begin{rem}
 According to Proposition ~\ref{pro2.2}, in the mean-uncertainty SVM method, when $\overline{\mu }=\underline{\mu }$, we have $\overline{P_i}=\underline{P_i}$, then the probability of $Y_i=1$ degenerate to $\sigma\Bigl( d(\boldsymbol{x_i}) \Bigr)$, where $d(x_{i})$ corresponds to the traditional SVM prediction.
\end{rem}

\section{Data}\label{sec:Data}

\subsection{Data Selection}\label{subsec:data_selection}

For the prediction of short-term average return directions in China's high frequency futures trading market, in order to more comprehensively investigate the predictability, we consider to choose different categories of commodity futures data. Since the futures data with high liquidity help enhancing the effectiveness of investment strategies, we prioritize futures contracts with higher liquidity. As is known, the most active contracts are typically with high liquidity and actively traded in the futures market. Therefore, we mainly focus on the most active contracts for different categories of future products. Specifically, we utilize turnover as a key metric to assess the market liquidity, enabling the ranking of futures products based on their respective turnover volumes. According to the turnovers in September 2024, the top 15 futures products are selected as the experimental samples, including Gold (AU), Copper (CU), Tin (SN), Nickel (NI), Corn (C), Zinc (ZN), Bottle Chips (PR), Rebar (RB), Silver (AG), Aluminum (AL), Hogs (LH), Cotton (CF), Lead (PB), Corn Starch (CS), and Soybean Meal (M). The details are shown in Figure~\ref{fig:sample}. The selected dataset consist of trading data and limit order book data which are obtained from the CTP-API and collected every 0.5 seconds. The time span ranges from October 1st 2024 to October 31st 2024, covering a total of 18 trading days.

Given that futures markets often experience sharp fluctuations during market opening and closing periods, these abnormal data points may seriously affect the stability and accuracy of predictive models. To mitigate the noise resulting from short-term market sentiment fluctuations, Data of the five-minute pre-opening and pre-closing periods are excluded.

\begin{figure}[htbp]
  \centering
  \includegraphics[width=1.0\textwidth]{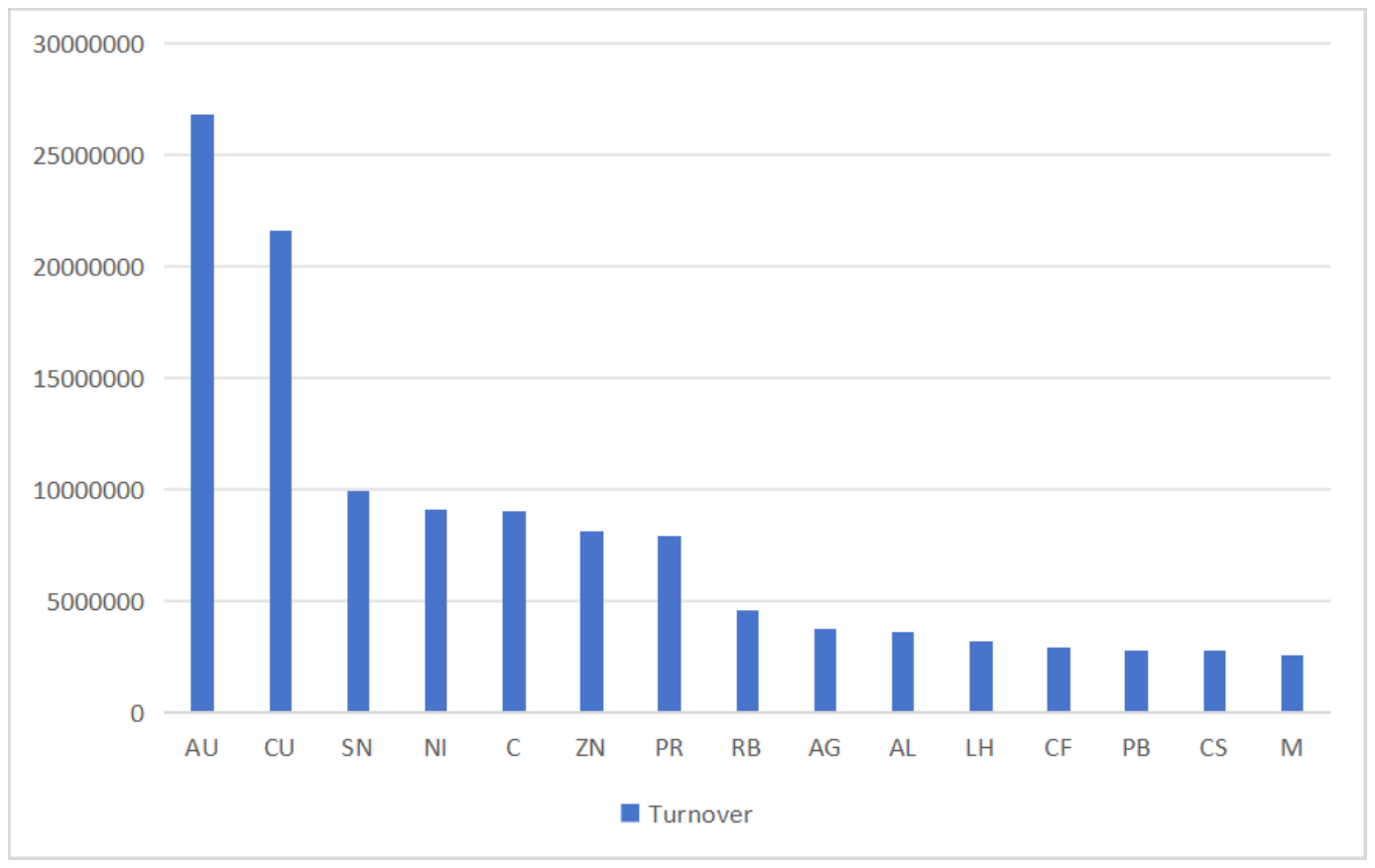} 
  \caption{Top 15 future products by turnover volume in September 2024}
  \label{fig:sample}
\end{figure}

\subsection{Dependent and Independent Variables}

Based on the factors proposed by Fan et al. (2022)\cite{Fan2022how} to characterize the original data, we first select the variables applicable to China's market, then refine the original independent variables to ensure compatibility with China-specific time-sliced data structures.

\subsubsection{Dependent Variable}

We utilize the binary classification label indicating the direction of short-term average returns as the dependent variable in the model. Here we formulate two distinct binary classification tasks for the prediction: the first predicts "upward movement" versus "non-upward" states, while the second targets "downward movement" versus "non-downward" states.

To define the dependent variable, we first introduce the concept of calendar time. For the trade and limit order book data, each trade's timestamp is recorded as \( t \in \mathbb{R}^+ \). Let \( V_t \) denote the trading volume at time \( t \), where \( V_t = 0 \) indicates that no trade occurred at time \( t \). The calendar time interval between \( T_1 \) and \( T_2 \), where \( T_1, T_2 \in \mathbb{R} \), can be defined as the following half-open interval:

\begin{equation*}
    Int\left( {{T}_{1}},{{T}_{2}} \right)=\{t\in \mathbb{R}:{{T}_{1}}<t\le {{T}_{2}}\}.
\end{equation*}

Then the definition of the forward-looking interval is given in the following:
\begin{equation*}
    In{{t}^{forward}}\left( T,\Delta  \right)=Int(T,T+\Delta ).
\end{equation*}
where $\Delta$ denotes the time range over which the prediction occurs. The interval is defined as a set of consecutive time stamps that are used to define the dependent variable.

Let \( D^{{txn}} \) denote the set of all record timestamps \( t \in \mathbb{R}^{+} \) at which the trades occur, and \( D^{\text{qt}} \) represent its quote counterpart. Set \( D = {{D}^{txn}} \cup {{D}^{qt}} \) . The best bid and ask prices at time \( t \in D \) are denoted by \( P_{t}^{b} \) and \( P_{t}^{a} \), respectively, and the mid-price is the average of these two prices, i.e., \( P_{t} = \frac{P_{t}^{b} + P_{t}^{a}}{2} \). Let \( P_{t}^{txn} \) denote the transaction price at time \( t \in {{D}^{txn}} \). The best bid size and best ask size are denoted by \( S_{t}^{b} \) and \( S_{t}^{a} \).

In the following, the definition of the short-term average return is provided. For \( T \in D \) and a short time horizon \( \Delta \), the short-term average return is defined as:
\begin{equation*}
    \text{Return}(T,\Delta )=\text{Average}\left[ P_{t}^{txn}:t\in {{D}^{txn}}\cap In{{t}^{forward}}(T,\Delta ) \right]/{{P}_{T}}-1.
\end{equation*}

The above formula represents the average return over a very short time span \( \Delta \). Comparing with  the returns of individual trades or fixed-interval, short-term average returns exhibit lower volatility and reflects more accurately the overall trading behavior within a \( \Delta \)-time frame.

After calculating the short-term average returns, we annotate the dataset with categorical labels. Essentially, minor fluctuations with limited amplitude and short durations are treated as market noise under our framework, as they fail to capture authentic market dynamics. To facilitate robust investment decision-making based on the classification outcomes, only very significant upward(or downward) movement are considered as valid fluctuation signals in our prediction.
Therefore, we construct classification datasets based on the 95th percentile threshold of short-term average returns for the two distinct classification tasks: for the "upward movement" versus "non-upward"  prediction task, observations exceeding this critical value are labeled "1" (minority class), while others receive "0" (majority class); similarly, for the "downward movement" versus "non-downward" task, observations falling below the 5th percentile threshold are labeled "1" (minority class), while all others are designated "0" (majority class).

\subsubsection{Independent Variables (Predictors)}

Similar to the forward-looking interval $In{{t}^{forward}}\left( T,\Delta  \right)$ for the dependent variable, A lookback interval is constructed to build the independent variables (or predictors). For the calendar time with timestamp \( T \), the lookback span is denoted as \( (\Delta_1, \Delta_2) \), where \( \Delta_1 \leq \Delta_2 \). Then the following lookback interval is defined:

\begin{equation*}
    In{{t}^{back}}\left( T,{{\Delta }_{1}},{{\Delta }_{2}} \right)=Int(T-{{\Delta }_{2}},T-{{\Delta }_{1}}).
\end{equation*}

For each time stamp \(T\), we use four look-back windows measured in seconds: $\left( {{\Delta }_{1}},{{\Delta }_{2}} \right)\in \left\{ \left( 0,2.5 \right),\left( 2.5,6.5 \right),\left( 6.5,12.5 \right),(12.5,25) \right\}$. Then we consider eight primary independent variables, each of which is computed over the four separate time spans. We category the eight independent variables in three groups, as shown in the following:

\textbf{Group I: Volume and duration.} The first group of predictors are related to the trading intensity of the future.

(1) Total volume (Volumeall): The total number of positions transacted in a given time interval,
 \begin{equation*}
     \text{VolumeAll}(T,{{\Delta }_{1}},{{\Delta }_{2}})=\sum\limits_{t\in In{{t}^{back}}(T,{{\Delta }_{1}},{{\Delta }_{2}})}{V_{t}}.
 \end{equation*}

 (2) Maximum volume (Volumemax): The maximum volume of a single transaction within a certain time interval,
 \begin{equation*}
    \text{VolumeMax}(T,{{\Delta }_{1}},{{\Delta }_{2}})=\max \left\{ {{V}_{t}}:t\in In{{t}^{back}}(T,{{\Delta }_{1}},{{\Delta }_{2}}) \right\}.
\end{equation*}

\textbf{Group II: Returns and imbalances.} The second group of predictors are related to the recent trading asymmetry of the stock.

(3) Price change (Lambda): The change in price relative to total volume over a certain time interval,
\begin{equation*}
    \text{Lambda}(T,{{\Delta }_{1}},{{\Delta }_{2}})=\frac{{{P}_{\max (\text{I})}}-{{P}_{\min (\text{I})}}}{\text{VolumeAll}(T,{{\Delta }_{1}},{{\Delta }_{2}})},
\end{equation*}
Where \(I={{D}^{txn}}\cap In{{t}^{back}}(T,\Delta_{1},\Delta_{2})=\{t \in \mathbb{R} : T-\Delta_{2} < t \le T-\Delta_{1}\}\), \({{P}_{\min (I)}}\) denotes the mid-price at the earliest time in interval \(I\), and \({{P}_{\max (I)}}\) denotes the mid-price at the latest time in interval \(I\).

(4) Quote Imbalance (LobImbalance): The average degree of imbalance in the depth of the Limit Order Book (LOB) over a certain time interval,
\begin{equation*}
    \text{LobImbalance}(T,{{\Delta }_{1}},{{\Delta }_{2}})=\text{Average}\left[ \frac{S_{t}^{a}-S_{t}^{b}}{S_{t}^{a}+S_{t}^{b}}:t\in In{{t}^{back}}(T,{{\Delta }_{1}},{{\Delta }_{2}}) \right],
\end{equation*}
Where $S_{t}^{b}$ is the quantity bought at the best bid price and $S_{t}^{a}$ is the quantity sold at the best ask price.

(5) Turnover Imbalance (TxnImbalance): A measure of the asymmetry between the buy and sell volumes in recent transactions,
\begin{equation*}
    \text{TxnImbalance}(T,{{\Delta }_{1}},{{\Delta }_{2}})=\frac{\sum\limits_{t\in {{D}^{txn}}\cap In{{t}^{back}}(T,{{\Delta }_{1}},{{\Delta }_{2}},M)}{({{V}_{t}}\cdot \text{Dir}_{\text{t}}^{\text{LR}})}}{\text{VolumeAll}(T,{{\Delta }_{1}},{{\Delta }_{2}})},
\end{equation*}
Where $Di{{r}^{LR}}$ is derived from the Lee-Ready \cite{lee1991inferring} algorithm, which infers the direction of a trade from a sequence of trades.

(6) Historical Return (PastReturn): The return within a certain time interval,
\begin{equation*}
    \text{PastReturn}(T,{{\Delta }_{1}},{{\Delta }_{2}})=1-\text{Average}\left[ P_{t}^{txn}:t\in I \right]/{{P}_{\max (I)}}.
\end{equation*}

\textbf{Group III: Speed and cost.} The third group of predictors focus on the speed and cost of futures trading.

(7) TurnOver: The speed of transactions in relation to the total number of positions,
\begin{equation*}
    \text{TurnOver}(T,{{\Delta }_{1}},{{\Delta }_{2}})=\frac{\text{VolumeAll}(T,{{\Delta }_{1}},{{\Delta }_{2}})}{S},
\end{equation*}
where $S$ refers to the total number of positions.

(8) Quoted Spread (QuotedSpread): The average proportional nominal spread of quotes over a certain time interval,
\begin{equation*}
    \text{QuotedSpread}(T,{{\Delta }_{1}},{{\Delta }_{2}})=\text{Average}\left[ \frac{P_{t}^{a}-P_{t}^{b}}{{{P}_{t}}}:t\in In{{t}^{back}}(T,{{\Delta }_{1}},{{\Delta }_{2}}) \right].
\end{equation*}

Finally, the eight factors were realized in each of the four look-back intervals. Therefore 32 factors are obtained and then used as inputs in the classification model.

\subsection{Training and Testing}

The model training and testing are performed with a rolling window of 3 trading days, where the data of the first 2 trading days is used for the training, and that of the 3rd  trading day is used for the testing, then the window rolls forward by 1 trading day to repeat the process. And the parameter ${{\hat{\overline{\mu }}}_{N}}$ in ~\eqref{upper P for mean-uncertainty LR} and ~\eqref{upper P for mean-uncertainty SVM} for the two mean-uncertainty methods are calculated during the training process. The datails for the training and test are illustrated in Figure~\ref{fig:sample}.

\begin{figure}[htbp]
  \centering
  \includegraphics[width=0.8\textwidth]{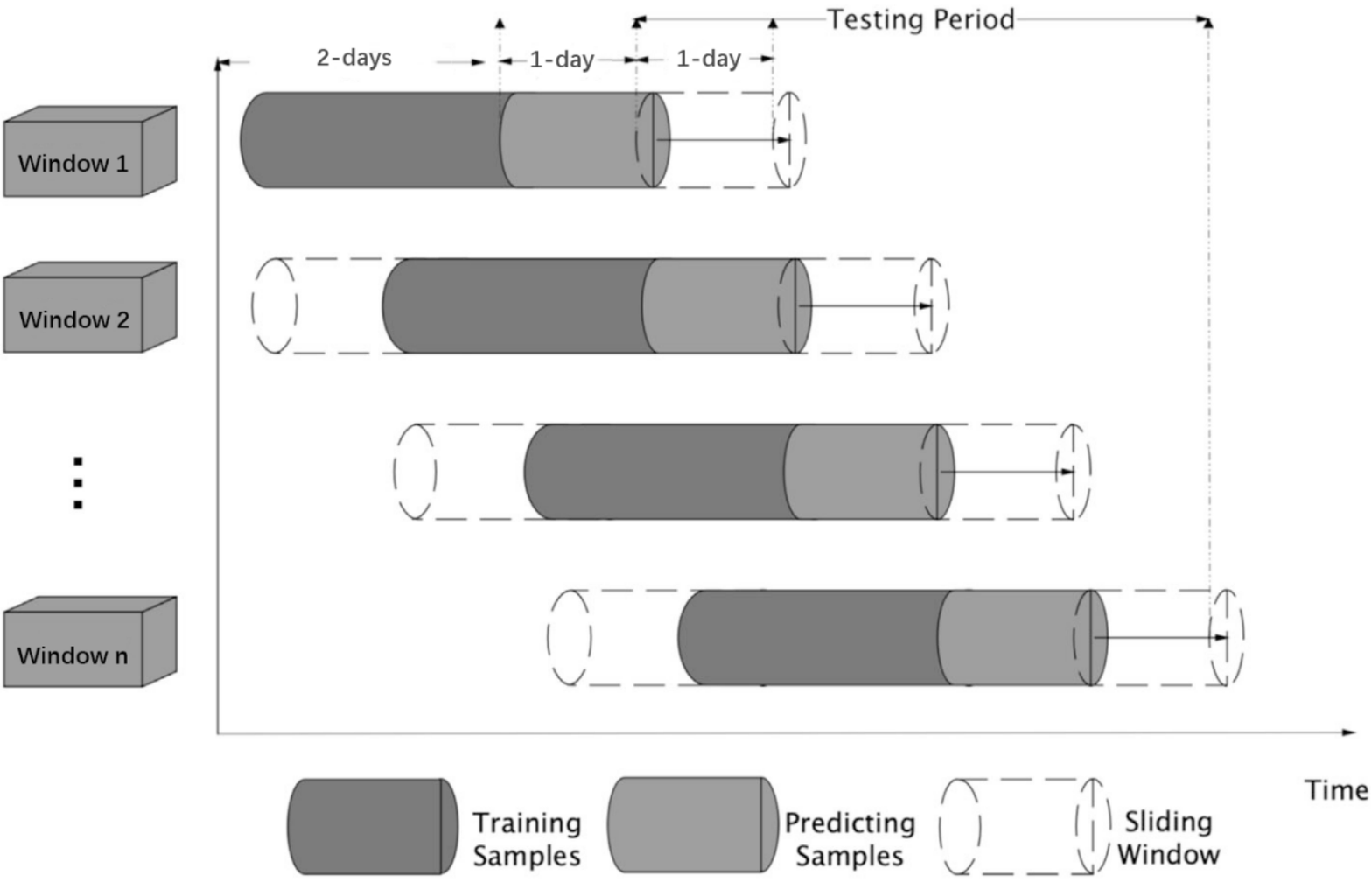} 
  \caption{Training and testing sliding window}
  \label{fig:sample}
\end{figure}

    In each trading day, we calculate the independent and dependent variables presented in Section 3.2, then label the dataset using the calculated dependent variable results. Moreover, a rolling time window with a duration of 30 seconds is employed within each trading day, where the data of the first 25 seconds is used to compute the independent variables, and that of the subsequent 5 second are used to calculate the dependent variable. The window rolls every 5 seconds to repeat the process.

\subsection{Overall Framework}
We show in Figure~ \ref{FIGhtmltext} the overall framework for the prediction of short-term average return directions with both the mean-uncertainty LR and the mean-uncertainty SVM methods. Moreover, we compare the prediction results with several conventional LR-related and SVM-related machine learning models (including the LR, SMOTE LR, RUS LR, SVM, SMOTE SVM, and RUS SVM methods) as shown in step 3 and 4 of Figure~\ref{FIGhtmltext}, and the details for the comparison are presented in Section ~\ref{sec:empirical_results}.
\begin{figure}[htbp]
\centering
\includegraphics[width=1.1\textwidth]{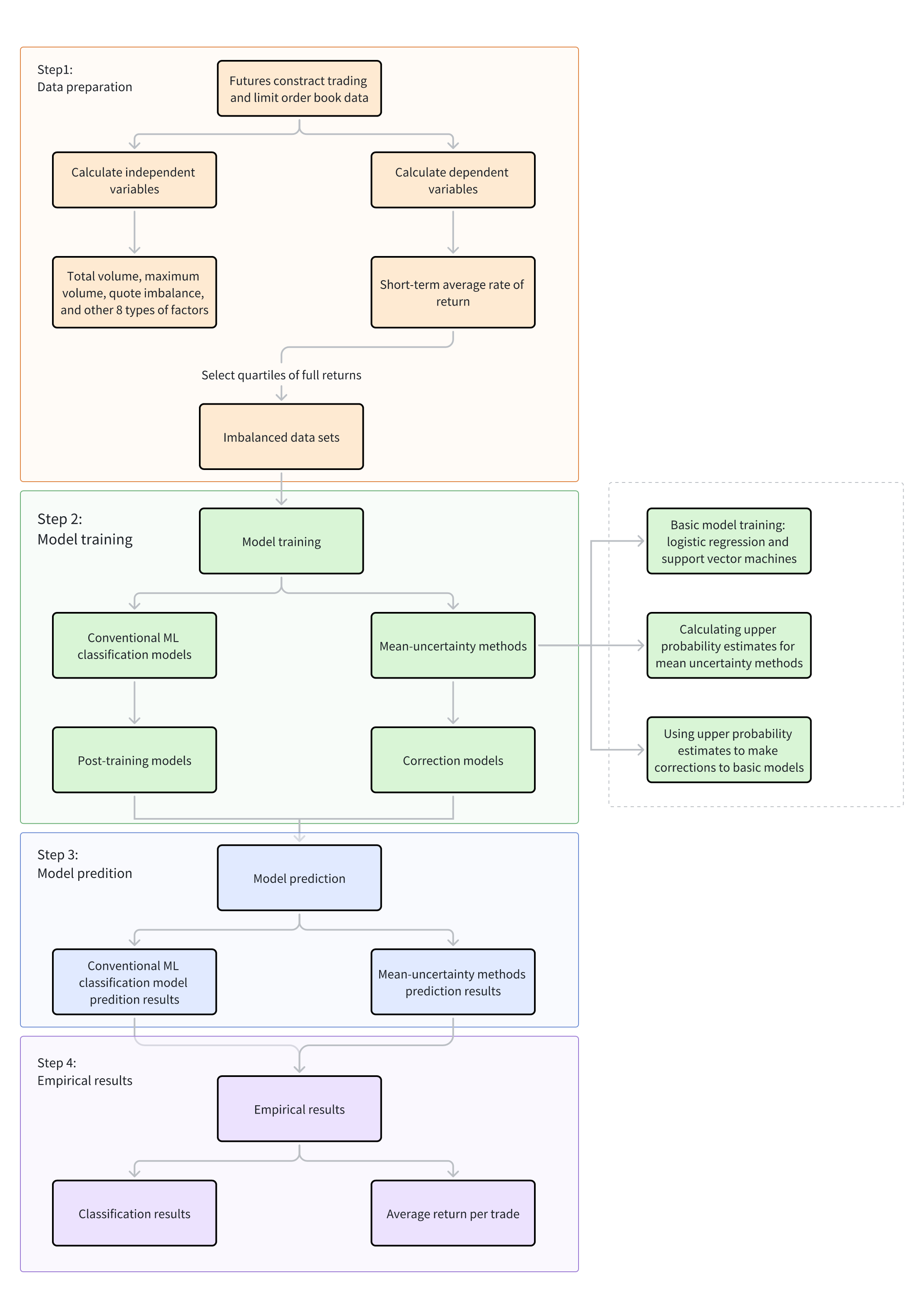}
\caption{Overall approach}
\label{FIGhtmltext}
\end{figure}

\section{Empirical Results}\label{sec:empirical_results}
\subsection{Classification Metrics}\label{subsec:Classification Metrics}

In order to assess the model performance of the prediction, the confusion matrix is first given in Table ~\ref{tab:diagbox}, then three key evaluation metrics suitable for imbalanced data classification are employed: Recall, Balanced Accuracy, and the $\beta$-variable $F$-score.

\begin{table}[htbp]
 \centering
  \caption{Confusion matrix}
  \label{tab:diagbox}
  \begin{tabular}{lcc}
    \toprule
    \diagbox{Actual}{Predictive} & Negative & Positive \\
    \midrule
    Negative & TN (True Negative) & FP (False Positive) \\
    Positive & FN (False Negative) & TP (True Positive) \\
    \bottomrule
  \end{tabular}
\end{table}

(1) Recall: The accuracy of a positive sample, i.e., the proportion of samples in which the actual sample category is positive and the prediction is also positive, defined as:
\begin{equation*}
    Recall=\frac{TP}{TP+FN}.
\end{equation*}

(2) Balance Accuracy: The average accuracy of the classification for two types of samples:
\begin{equation*}
    Bacc=\frac{1}{2}(\frac{TP}{TP+FN}+\frac{TN}{TN+FP}),
\end{equation*}
where $\frac{TP}{TP+FN}$($Recall$) is the accuracy of the positive samples and $\frac{TN}{TN+FP}$ is the accuracy of the negative samples.

(3) $\beta $-Varied $F$-Measure:
\begin{equation*}
    {{F}_{\beta }}=\frac{(1+{{\beta }^{2}})\cdot Recall\cdot Precision}{Recall+{{\beta }^{2}}\cdot Precision},
\end{equation*}
where $Precision$ refers to the proportion of the actual positive samples among the samples predicted as positive,
\begin{equation*}
    Precision = \frac{TP}{TP+FP}
\end{equation*}

  Here \( TP \) and \( FP \) represents respectively true positives and false positives.
  The coefficient \( \beta \) in $\beta $-Varied $F$-Measure is used to adjust the relative importance between Precision and Recall. In the context of imbalanced data, \( \beta \) is typically set as \( \ln \frac{m_{-}}{m_{+}} \), where \( m_{-} \) and \( m_{+} \) denote the number of negative and positive samples, respectively.

\subsection{Classification Results}

Based on the high-frequency data from China's futures market, we employ the two mean-uncertainty classification methods, mean-uncertainty LR and mean-uncertainty SVM, to predict the direction of short-term average returns for the futures contracts. For the implementation of the mean-uncertainty SVM method, we utilize the Gaussian kernel as the kernel function $K(\boldsymbol{x}, \boldsymbol{\tilde{{x}}_{j}})$ in ~\eqref{1}, such as \(K(\boldsymbol{x}, \boldsymbol{\tilde{{x}}_{j}}) = \exp \left( -\gamma \left\| \boldsymbol{x} - \boldsymbol{\tilde{{x}}_{j}} \right\|^2 \right),\) where the super parameter $\gamma$ can be determined based on the number of features and the variance of the training data. By using the metrics presented in Section ~\ref{subsec:Classification Metrics}, the classification performance of the two mean-uncertainty methods is then compared with that of several conventional LR-related and SVM-related models, including the LR, SMOTE-LR, RUS-LR, SVM, SMOTE-SVM, and RUS-SVM methods.

Without loss of generality, we illustrate the classification results for the most active contracts for Gold (AU), Tin (SN), Zinc (ZN), and Rebar (RB), all of which rank among the top 15 futures products by market liquidity. The prediction task focuses on the direction of "upward" or "non-upward" over the period from October 1 to October 30, 2024, encompassing a total of 18 trading days.

We first show in Table~\ref{tab:lr_eval} the comparison results between the mean-uncertainty LR method and another
three conventional LR-related machine learning methods, including the LR, SMOTE-LR and RUS-LR methods. From Table~\ref{tab:lr_eval}, we can see that the mean-uncertainty LR method outperforms other methods across all the three metrics, with a particularly notable advantage in Recall. It also achieves the highest Balanced Accuracy, indicating better sensitivity to minority class identification. In contrast, the basic LR method produces near-zero Recall on the testing sets, which demonstrate its ineffectivity at detecting minority class samples. Although the SMOTE-LR and RUS-LR demonstrate some improvements Comparing with the LR, the classification performance are still inferior to that of the mean-uncertainty LR method.

\begin{table}[H]
  \centering
  \caption{Classification results of LR-related models}
  \label{tab:lr_eval}
  \begin{tabular}{llcccc}
    \toprule
   Varieties & Metrics & LR & SMOTE-LR & RUS-LR & Mean-uncertainty LR \\
    \midrule
    \multirow{3}{*}{\begin{tabular}{c} AU \end{tabular}}
      & Bacc      & 0.504 & 0.576 & 0.503 & \textbf{0.577} \\
      & F-Measure & 0.010 & 0.331 & 0.266 & \textbf{0.348} \\
      & Recall    & 0.009 & 0.596 & 0.505 & \textbf{0.697} \\
    \midrule
    \multirow{3}{*}{\begin{tabular}{c} SN \end{tabular}}
      & Bacc      & 0.500 & 0.550 & 0.574 & \textbf{0.609} \\
      & F-Measure & 0.000 & 0.281 & 0.315 & \textbf{0.354} \\
      & Recall    & 0.000 & 0.459 & 0.532 & \textbf{0.615} \\
    \midrule
    \multirow{3}{*}{\begin{tabular}{c} ZN \end{tabular}}
      & Bacc      & 0.500 & 0.577 & 0.575 & \textbf{0.622} \\
      & F-Measure & 0.000 & 0.316 & 0.317 & \textbf{0.364} \\
      & Recall    & 0.000 & 0.529 & 0.539 & \textbf{0.629} \\
    \midrule
    \multirow{3}{*}{\begin{tabular}{c} RB \end{tabular}}
      & Bacc      & 0.504 & 0.608 & 0.609 & \textbf{0.640} \\
      & F-Measure & 0.010 & 0.349 & 0.357 & \textbf{0.368} \\
      & Recall    & 0.009 & 0.557 & 0.600 & \textbf{0.610} \\
    \bottomrule
  \end{tabular}
\end{table}

\begin{table}[H]
  \centering
  \caption{Classification results of SVM-related models }
  \label{tab:svm_eval}
  \begin{tabular}{llcccc}
    \toprule
     Varieties & Metrics & SVM & SMOTE-SVM & RUS-SVM & Mean-uncertainty SVM \\
    \midrule
    \multirow{3}{*}{\begin{tabular}{c} AU \end{tabular}}
      & Bacc      & 0.500 & 0.551 & 0.494 & \textbf{0.571} \\
      & F-Measure & 0.000 & 0.228 & 0.228 & \textbf{0.361} \\
      & Recall    & 0.000 & 0.294 & 0.385 & \textbf{0.817} \\
    \midrule
    \multirow{3}{*}{\begin{tabular}{c} SN \end{tabular}}
      & Bacc      & 0.500 & 0.526 & \textbf{0.537} & 0.523 \\
      & F-Measure & 0.000 & 0.190 & 0.228 & \textbf{0.370} \\
      & Recall    & 0.000 & 0.245 & 0.318 & \textbf{0.866} \\
    \midrule
    \multirow{3}{*}{\begin{tabular}{c} ZN \end{tabular}}
      & Bacc      & 0.500 & 0.522 & 0.582 & \textbf{0.595} \\
      & F-Measure & 0.000 & 0.134 & 0.286 & \textbf{0.314} \\
      & Recall    & 0.000 & 0.170 & 0.604 & \textbf{0.738} \\
    \midrule
    \multirow{3}{*}{\begin{tabular}{c} RB \end{tabular}}
      & Bacc      & 0.500 & 0.518 & \textbf{0.551} & 0.538 \\
      & F-Measure & 0.000 & 0.135 & 0.252 & \textbf{0.325} \\
      & Recall    & 0.000 & 0.172 & 0.469 & \textbf{0.679} \\
    \bottomrule
  \end{tabular}
\end{table}

Table~\ref{tab:svm_eval} shows the comparison results between the mean-uncertainty SVM method and another
three conventional SVM-related machine learning methods, including the SVM, SMOTE-SVM and RUS-SVM methods. We can see that the basic SVM method can hardly identify minority class samples, highlighting its inefficiency in handling imbalanced data. In contrast, the mean-uncertainty SVM method outperforms the others in both the $F$-Measure and Recall. Particularly, it shows a significant improvement in Recall over all the three conventional SVM-related methods. This demonstrates that mean-uncertainty SVM is more effective in recognizing minority class samples. Although SMOTE-SVM and RUS-SVM exhibit improvements relative to the basic SVM, their overall performance on the testing set remains inferior to that of the mean-uncertainty SVM.

\subsection{Investment Strategy and the Average Return per Trade}

To assess the practical effectiveness of the mean-uncertainty methods in China's high-frequency futures market, we develop investment strategies based on the prediction results from the two mean-uncertainty classification models (mean-uncertainty LR and mean-uncertainty SVM). Specifically, for the two distinct prediction tasks, the "upward/non-upward" prediction corresponds to a "long/not long" strategy, while the "downward/non-downward" prediction corresponds to a "short/not short" strategy. Then we compute the average return per trade for each method and compare the computation results.

Taking the "long/not long" strategy (corresponding to "upward/non-upward" prediction) as an example, the cumulative return  over 18 trading days are first calculated:

\begin{equation}\label{eq:Cumulative Return}
    \text{Cumulative Return}=\prod\limits_{i=1}^{n}{\left( 1+{{r}_{i}} \right)-1},
\end{equation}
where \( n \) represents the number of minority class samples predicted over 18 trading days, i.e., the number of trades. \( r_{i} \) denotes the short-term average return corresponding to the \( i \)-th sample predicted to be in the minority class, with the calculation of short-term average return presented in Section 3.2.1. The average return per trade is then calculated as:
\begin{equation}\label{eq:Average Return per Trade}
 \text{Average Return per Trade}=\frac{ \text{Cumulative Return}}{n}.
\end{equation}

As for the "short/not short" strategy, the minority class samples result in negative short-term average returns, which means that we can earn profits by short-selling these instances. In contrast to the long strategy, the short strategy profits from a decline in the underlying asset's price. Consequently, the cumulative return is calculated as the negative of the total return, then the average return per trade is determined by ~\eqref{eq:Average Return per Trade}.

After predicting the short-term average return directions across the top 15 futures contracts, the average returns per trade from the respective trading strategies are then summarized. The comparison results for the "upward/non-upward" (long/not long) predictions among different models are presented in Table~\ref{tab:table44}, and those for the "downward/non-downward" (short/not short) predictions are illustrated in Table~\ref{tab:table45}.

As is seen from Table~\ref{tab:table44}, for the "upward/non-upward" prediction, both the mean-uncertainty LR and the mean-uncertainty SVM methods outperform the other conventional methods in terms of average returns per trade across most futures contracts. Although the RUS or SMOTE based method yields some improvements comparing with the basic LR and SVM method, their performance remains inferior to that of the mean-uncertainty methods. Similar with the results in Table~\ref{tab:table44}, Table~\ref{tab:table45} illustrate that for the "downward/non-downward" prediction, the mean-uncertainty methods continuing to exhibit superior performance for the majority of contracts.

Overall, from the perspective of returns for the corresponding investment strategy, compared to the other conventional classification models in Tables ~\ref{tab:table44} and ~\ref{tab:table45}: for the "upward/non-upward" scenario, both the mean-uncertainty LR and the mean-uncertainty SVM methods achieve higher average returns per trade in 80\% of the futures contracts; for the "downward/non-downward" scenario, the mean-uncertainty LR method outperforms in 80\% of the LR-related futures contracts, and the mean-uncertainty SVM method performs better in 67\% of the SVM-related contracts. These results indicate that the mean-uncertainty methods demonstrate stronger advantages and more accurate predictive capabilities when addressing the data imbalance in the direction prediction of high-frequency futures data, which is more helpful for investment decisions.

\begin{landscape} 
\vspace*{\fill} 
  \begin{table}[htbp]
    \centering
    \caption{Average returns per trade on the test set for "upward/non-upward" predictions}
    \label{tab:table44}
    \begin{tabular}{lcccccccc}
      \toprule
      Varieties & \multicolumn{4}{c}{LR-related classification models} & \multicolumn{4}{c}{SVM-related classification models} \\
           & LR & SMOTE-LR & RUS-LR & Mean-uncertainty LR & SVM & SMOTE-SVM & RUS-SVM & Mean-uncertainty SVM \\
      \midrule
      AU   & -3.23e-6 & 1.26e-6 & 9.55e-8 & \textbf{1.31e-6} & 0      & 1.15e-6 & 1.21e-6 & \textbf{3.63e-6} \\
      CU     & -1.76e-5 & 1.89e-5 & 1.31e-5 & \textbf{2.25e-5} & 0      & 1.58e-5 & 1.36e-5 & \textbf{2.83e-5} \\
      SN     & -3.45e-5 & 1.11e-5 & 7.97e-6 & \textbf{1.14e-5} & 0      & 6.72e-6 & 6.15e-6 & \textbf{1.56e-5} \\
      NI     & 2.95e-6  & 6.98e-6 & 3.74e-6 & \textbf{8.58e-6} & 0      & 2.81e-6 & 2.09e-6 & \textbf{3.75e-6} \\
      C   & 1.02e-4  & 1.07e-4 & 8.22e-5 & \textbf{1.15e-4} & 0      & 7.48e-5 & 8.49e-5 & \textbf{1.17e-4} \\
      ZN     & 5.93e-6  & 3.45e-5 & 2.80e-6 & \textbf{3.68e-5} & 0      & 2.33e-5 & 2.35e-5 & \textbf{4.00e-5} \\
      PR   & 4.14e-5  & \textbf{5.17e-5} & 2.66e-5 & 4.76e-5 & 0      & 3.55e-5 & 3.02e-5 & \textbf{4.86e-5} \\
      RB & 3.33e-5  & 6.04e-5 & 4.61e-5 & \textbf{6.14e-5} & 0      & 3.21e-5 & 4.48e-5 & \textbf{4.76e-5} \\
      AG     & 7.65e-6  & \textbf{1.04e-5} & 7.86e-6 & 9.47e-6 & 0      & 3.53e-6 & 9.22e-6 & \textbf{1.45e-5} \\
      AL     & 4.28e-5  & 6.44e-5 & 5.59e-5 & \textbf{6.63e-5} & 0      & 4.81e-5 & 5.51e-5 & \textbf{6.56e-5} \\
      LH   & 1.26e-5  & 7.56e-5 & 5.62e-5 & \textbf{7.90e-5} & 0      & 4.97e-5 & \textbf{5.54e-5} & 5.53e-5 \\
      CF   & 7.91e-5  & 8.54e-5 & 7.01e-5 & \textbf{8.95e-5} & 0      & 5.73e-5 & 6.87e-5 & \textbf{8.34e-5} \\
      PB     & -4.82e-5 & 9.11e-5 & 6.95e-5 & \textbf{9.64e-5} & 0      & 6.05e-5 & \textbf{7.93e-5} & 5.68e-5 \\
      CS & 1.48e-5 & 1.14e-4 & 8.08e-5 & \textbf{1.14e-4} & 0      & 9.20e-5 & 7.91e-5 & \textbf{1.23e-4} \\
      M   & 3.34e-5  & \textbf{9.17e-5} & 7.97e-5 & 8.57e-5 & 0      & 6.34e-5 & \textbf{7.14e-5} & 5.64e-5 \\
      \bottomrule
    \end{tabular}
  \end{table}
  \vspace*{\fill} 
\end{landscape} 

\begin{landscape}
\vspace*{\fill} 
\begin{table}[htbp]
  \centering
  \caption{Average returns per trade on the test set for "downward/Non-downward" predictions}
  \label{tab:table45}
  \begin{tabular}{lcccccccc}
    \toprule
    Varieties & \multicolumn{4}{c}{LR-related classification models} & \multicolumn{4}{c}{SVM-related classification models} \\
         & LR & SMOTE-LR & RUS-LR & Mean-uncertainty LR & SVM & SMOTE-SVM & RUS-SVM & Mean-uncertainty SVM \\
    \midrule
    AU   & -2.34e-05 & \textbf{2.05e-06} & 7.58e-07 & 1.30e-06 & 0 & 1.07e-06 & 7.62e-07 & \textbf{2.68e-06} \\
    CU     & 3.30e-06  & 2.39e-05 & 2.03e-05 & \textbf{2.62e-05} & 0 & 2.27e-05 & 1.94e-05 & \textbf{3.38e-05} \\
    SN     & 1.94e-06  & 1.61e-05 & 1.52e-05 & \textbf{1.86e-05} & 0 & 1.11e-05 & 1.26e-05 & \textbf{3.09e-05} \\
    NI     & -9.81e-05 & 8.10e-06 & 7.05e-06 & \textbf{8.30e-06} & 0 & \textbf{8.46e-06} & 6.70e-06 & 7.24e-06 \\
    C   & 1.50e-05  & 5.15e-05 & 4.53e-05 & \textbf{5.68e-05} & 0 & 5.18e-05 & 4.17e-05 & \textbf{6.21e-05} \\
    ZN     & 5.60e-06  & \textbf{2.46e-05} & 2.17e-05 & 2.41e-05 & 0 & 1.70e-05 & 1.69e-05 & \textbf{2.45e-05} \\
    PR   & 1.91e-05  & 4.28e-05 & 3.31e-05 & \textbf{4.55e-05} & 0 & \textbf{3.52e-05} & 3.19e-05 & 2.34e-05 \\
    RB & 6.89e-07  & 3.00e-05 & 2.68e-05 & \textbf{3.01e-05} & 0 & 2.08e-05 & \textbf{2.18e-05} & 1.16e-05 \\
    AG     & 7.82e-07  & \textbf{7.48e-06} & 5.42e-06 & 6.63e-06 & 0 & 5.76e-06 & 6.61e-06 & \textbf{1.81e-05} \\
    AL     & 9.42e-06  & 4.27e-05 & 3.89e-05 & \textbf{4.79e-05} & 0 & \textbf{4.40e-05} & 3.78e-05 & 3.98e-05 \\
    LH   & -6.05e-06 & 4.42e-05 & 3.57e-05 & \textbf{4.67e-05} & 0 & 4.44e-05 & 3.51e-05 & \textbf{5.78e-05} \\
    CF   & 9.55e-06  & 4.75e-05 & 4.29e-05 & \textbf{4.84e-05} & 0 & \textbf{4.93e-05} & 3.83e-05 & 4.56e-05 \\
    PB     & 4.54e-05  & 5.82e-05 & 5.08e-05 & \textbf{6.46e-05} & 0 & 5.96e-05 & 4.97e-05 & \textbf{7.30e-05} \\
    CS & 1.14e-05 & 4.33e-05 & 5.14e-05 & \textbf{5.46e-05} & 0 & 5.49e-05 & 4.16e-05 & \textbf{6.80e-05} \\
    M   & 5.71e-06  & 4.40e-05 & 3.84e-05 & \textbf{4.71e-05} & 0 & 4.35e-05 & 3.54e-05 & \textbf{5.37e-05} \\
    \bottomrule
  \end{tabular}
\end{table}
\vspace*{\fill} 
\end{landscape}

\section{Conclusion}

This paper investigates the predictability of short-term average return directions using high-frequency data from China's futures market, and facilitate corresponding investment decision-making. Addressing the issue of market data imbalance, we extend the mean-uncertainty LR method under the SLE framework and develop the mean-uncertainty SVM method. Employing both the mean-uncertainty methodologies, we first implement binary classification and develop corresponding investment strategies, then perform empirical analysis on the top 15 liquid products among the most active contracts in China's future market. Compared with traditional classification methods, both the mean-uncertainty classifier achieves higher average returns per trade on 80\% of futures products in "upward vs. non-upward" scenarios. For "downward vs. non-downward" predictions, the mean-uncertainty LR method delivers superior returns on 80\% of products among LR-related methods, and the mean-uncertainty SVM outperforms on 67\% of products among SVM-related methods.

\appendix

   \section{Preliminaries of Sub Linear Expectation (SLE)}\label{appendix:SLE}
   In the appendix section, we introduce some useful definitions and theorems of SLE, which are the basis of the mean-uncertainty LR and mean-uncertainty SVM classification method. More theoretical results and proofs for sublinear expectation can be referred to \cite{peng2010nonlinear}.

Consider a sublinear expectation space  $(\Omega, \mathcal{H}, {\mathbb{E}})$, where $\Omega$ is a given non-empty set space, $\mathcal{H}$ is a space of real-valued functions(or a space of random variables) defined on $\Omega$. Without loss of generality, assume that the space of random variables $\mathcal{H}$ satisfies the following condition: if $X \in \mathcal{H}$, then for any function $\varphi \in \mathcal{C}_{\text{Lip}}(\mathbb{R}^d)$, we have $\varphi(X) \in \mathcal{H}$. Here, $\mathcal{C}_{\text{Lip}}(\mathbb{R}^d)$ denotes the linear space of all real-valued Lipschitz continuous functions defined on $\mathbb{R}^d$. When these functions are additionally bounded, they form the space $\mathcal{C}_{b.\text{Lip}}(\mathbb{R}^d)$. $\widehat{\mathbb{E}}$ is a sublinear expectation on $(\Omega, \mathcal{H}, \widehat{\mathbb{E}})$ which is defined as follows:

\begin{defn}\label{def2.1}
(Sublinear expectation) The sublinear expectation $\mathcal{H}$ is a functional mapping $\widehat{\mathbb{E}}: \mathcal{H} \to \mathbb{R}$ that satisfies the following properties for any given $X, Y \in \mathcal{H}$:

(1) Monotonicity: If $X(\omega) \ge Y(\omega)$ for all $\omega \in \Omega$, then $\widehat{\mathbb{E}}[X] \ge \widehat{\mathbb{E}}[Y]$;

(2) Constant preserving: $\widehat{\mathbb{E}}[c] = c$, $\forall c \in \mathbb{R}$;

(3) Sub-additivity: $\widehat{\mathbb{E}}[X + Y] \le \widehat{\mathbb{E}}[X] + \widehat{\mathbb{E}}[Y]$;

(4) Positive homogeneity: $\widehat{\mathbb{E}}[\lambda X] = \lambda \widehat{\mathbb{E}}[X], \quad \forall \lambda \ge 0$.

And $(\Omega, \mathcal{H}, \widehat{\mathbb{E}})$ is called a sublinear expectation space.
   \end{defn}

The sublinear expectation space encompasses all classical linear distributions, as well as a wide range of nonlinear distributions. Similar to the concepts in classical probability spaces, the maximum distribution can also be defined within the sublinear expectation space.

\begin{defn}\label{def2.2}
(Maximal distribution) A \(d\)-dimensional random variable \(\xi\) in sublinear space \(\left( \Omega, \mathcal{H}, \widehat{\mathbb{E}} \right)\), is maximally distributed if there exists a bounded, closed and convex subset $ \overline{\Theta}$ in $R^{d}$ such that

\[
\widehat{\mathbb{E}}[\varphi(\xi)] = \underset{v \in \overline{\Theta}}{\max} \, \varphi(v), \quad \varphi \in \mathcal{C}_{Lip}({\mathbb{R}}^d),
\]
which is denoted as \(\xi \overset{d}{\mathop{=}}\,M({\overline{\Theta}})\). If \(\overline{\Theta}\) is a non-convex set, it is referred to as the non-convex maximum distribution.
\end{defn}

Sublinear expectation plays a crucial role in the nonlinear framework, and sublinear expectations are equivalent to taking expectations about a family of linear expectations. Therefore, the sublinear expectation $\widehat{\mathbb{E}}$ can be used to equivalently characterize a family of uncertain probabilities ${{\left\{ {{P}_{\theta }} \right\}}_{\theta \in \Theta }}$.

\begin{thm}\label{the2.1}
Let \(\widehat{\mathbb{E}}\) be a sublinear expectation defined on \(\left( \Omega, \mathcal{H} \right)\), $\boldsymbol{X}=\left( {{X}_{1}},{{X}_{2}},\cdots ,{{X}_{n}} \right)$. Then, there exists a family of linear expectations \(\left\{ E_{\theta} : \theta \in \Theta \right\}\) defined on \(\left( \Omega, \mathcal{H} \right)\) such that:

\[
\widehat{\mathbb{E}}[X] = \underset{\theta \in \Theta}{\max} E_{\theta}[X].
\]
\end{thm}

\begin{defn}\label{def2.3}
(Independent and identically distributed) If for each \(i = 1, 2, \dots\), we have \(X_{i+1} \overset{d}{\mathop{=}} X_i\), and \(X_{i+1}\) is independent of \(\{X_1, \cdots, X_i\}\), then the sequence of \(\mathbb{R}^d\)-valued random variables \(\{X_i\}_{i=1}^{\infty}\) in the sublinear expectation space \(\left( \Omega, \mathcal{H}, \widehat{\mathbb{E}} \right)\) is called independent and identically distributed (i.i.d.).
\end{defn}

\begin{thm}\label{the2.2}
(Non-linear law of large number) Let \(\{ X_i \}_{i=1}^{\infty}\) be a sequence of i.i.d. \(\mathbb{R}^d\)-valued random variables in the sublinear expectation space \(\left( \Omega, \mathcal{H}, \widehat{\mathbb{E}} \right)\), and assume that

\[
\lim_{c \to \infty} \widehat{\mathbb{E}}\left[ \left( |X_1 - c| \right)^+ \right] = 0,
\]
then the sequence \(\left\{ \sum_{i=1}^{n} \frac{X_i}{n} \right\}_{i=1}^{\infty}\) converges in distribution to a maximum distribution:

\[
\lim_{n \to \infty} \widehat{\mathbb{E}}\left[ \varphi \left( \sum_{i=1}^{n} \frac{X_i}{n} \right) \right] = \max_{\mu \in \Gamma} \left[ \varphi(\mu) \right].
\]

For all functions \(\varphi \in \mathcal{C}({\mathbb{R}}^d)\) that satisfy the linear growth condition, \(\Gamma\) is the bounded convex closed set in \({\mathbb{R}}^d\) determined by the following:

\[
\max_{\mu \in \Gamma} \langle \mu, p \rangle = \widehat{\mathbb{E}}[\langle p, X_1 \rangle], \quad p \in {\mathbb{R}}^d.
\]

In particular, when \(d = 1\), we have \(\Gamma = [\overline{\mu}, \underline{\mu}]\), where:

\[
\overline{\mu} = \widehat{\mathbb{E}}[X_i], \quad \underline{\mu} = -\widehat{\mathbb{E}}[-X_i],
\]
which characterizes the uncertainty in the mean of the random variable.
\end{thm}

\begin{thm}\label{the2.3}
    ($\varphi$-max-min method)
    Let $X_{1},\dots ,X_{n}$ be a $1$-dimensional independent identically distributed sequence in linear space obeying a maximal distribution:
    $X_{i}\overset{d}{=}M_{\left [ \overline{\mu},\underline{\mu} \right ]}, i=1,\dots ,n,$
    where $\underline{\mu} \le \overline{\mu}$ are the two unknown parameters, then for any $P_{\theta},\theta \in \left [ \underline{\mu},\overline{\mu} \right ]$ the following relation holds almost everywhere:
    \begin{equation*}
        \underline{\mu}\le \min\left \{ X_{1}\left ( \omega  \right ) ,\dots,X_{n}\left ( \omega  \right )  \right \} \le \max\left \{ X_{1}\left ( \omega  \right ), \dots,X_{n}\left ( \omega  \right )   \right \} \le \overline{\mu}.
    \end{equation*}
    Further there are: $ \hat{\overline{\mu}}_{n} = \max\left \{ X_{1},\dots,X_{n} \right \} $ is the maximum unbiased estimate of the upper mean $\overline{\mu}$, $ \hat{\underline{\mu}}_{n} = \min\left \{ X_{1},\dots,X_{n} \right \} $ is the minimum unbiased estimate of the lower mean $\underline{\mu}$.
\end{thm}




\bibliographystyle{ieeetr}
\bibliography{references}

\end{document}